\documentclass[11pt,margin=1in,lmargin=1.25in]{amsart}

\usepackage{amsfonts, amsmath, amssymb, color}
\usepackage{graphicx}
\usepackage{amsthm}
\usepackage{bbm}
\usepackage{booktabs}
\usepackage{float}
\usepackage{hhline}
\usepackage{lipsum}
\usepackage{lscape}
\usepackage{subfig}
\usepackage{textcomp}
\usepackage{tikz}
\usetikzlibrary{decorations.pathmorphing,shapes,arrows,positioning}
\usepackage{xcolor}
\usepackage{young}
\usepackage{youngtab}
\usepackage{ytableau}
\usepackage{tikz}
\usepackage{float}
\usepackage{scalefnt}
\usetikzlibrary{calc}

\allowdisplaybreaks[4]
\theoremstyle{plain}
\newtheorem{thm}{Theorem}[section]
\newtheorem{lem}[thm]{Lemma}

\newtheorem{cor}[thm]{Corollary}

\theoremstyle{definition}

\newtheorem{exmp}[thm]{Example}

\newcommand{\Rmnum}[1]{\expandafter\@slowromancap\romannumeral #1@}

\numberwithin{equation}{section} \errorcontextlines=0
\begin{document}
\title{Weighted monogamy and polygamy relations
}
\author{Yue Cao}
\address{School of Mathematics, South China University of Technology,
Guangzhou, Guangdong 510640, China}
\email{434406296@qq.com}
\author{Naihuan Jing*}
\address{Department of Mathematics, North Carolina State University, Raleigh, NC 27695, USA}
\email{jing@ncsu.edu}
\author{Yiling Wang}
\address{Department of Mathematics, North Carolina State University, Raleigh, NC 27695, USA}
\email{ywang327@ncsu.edu}
\subjclass[2010]{Primary: 85; Secondary:}
\thanks{$*$Corresponding author: jing@ncsu.edu}
\maketitle
\begin{abstract}
This research offers a comprehensive approach to strengthening both monogamous and polygamous relationships within the context of quantum correlations in multipartite quantum systems. We pre\-sent {\it the most stringent} bounds for both monogamy and polygamy in multipartite systems compared to recently established relations. We show that whenever a bound is given (named it monogamy or polygamy), our bound indexed by some parameter $s$ will always be stronger than the given bound. The study includes detailed examples, highlighting that our findings exhibit greater strength across all existing cases in comparison.
\end{abstract}

\section{\textbf{Introduction}}
Quantum entanglement holds great significance in the realm of quantum information processing. Unlike classical correlation, where information can be freely shared among multiple parties, quantum entanglement imposes constraints in multipartite systems. In these quantum systems, the entanglement shared between one subsystem and another restricts the degree to which it can be entangled with additional systems \cite{ASI}. This limitation is referred to as the monogamy of quantum entanglement.

The concept of monogamy in quantum entanglement has various crucial applications in physics, particularly in the field of quantum information theory. It underscores the idea that entanglement is not distributable without limitations, leading to important implications for the manipulation and transmission of quantum information within multipartite quantum systems \cite{AMG,P}.

The first monogamy relation was discovered by Coffman, Kundu, and Wootters \cite{CKW} for concurrence in three-qubit states $\rho_{ABC}$:
\begin{equation}\label{e:mono1}
C^{2}(\rho_{A|BC})\geq C^{2}(\rho_{AB})+C^{2}(\rho_{AC}),
\end{equation}
where $\rho_{AB}$ and $\rho_{AC}$ are the reduced density matrices of $\rho_{ABC}$.
Osborne and Verstraete \cite{OV} generalized it to
$N$-qubit states: $$
C^{2}(\rho_{A|B_{1}\cdots B_{N-1}})\geq \sum_{i=1}^{N-1} C^{2}(\rho_{AB_{i}}).$$
Since then, the monogamy of different entanglements has been widely studied \cite{G, JHC,JLLF,ZF, KPSS, KDS,OF}. The generalized monogamy relation for arbitrary dimensional tripartite states was also established \cite{OV,BYW,SPSS}.

Polygamy of entanglement can be viewed as a dual form of monogamy. Gour et al observed that the assisted entanglement \cite{GMS} obeys:
\begin{equation}\label{e:poly1}
C^{2}_{a}(\rho_{A|BC})\leq C^{2}_{a}(\rho_{AB})+C^{2}_{a}(\rho_{AC}),
\end{equation}
and for $N$-qubit states \cite{GBS}:
$
C^{2}_{a}(\rho_{A|B_{1}\cdots B_{N-1}})\leq \sum_{i=1}^{N-1} C^{2}_{a}(\rho_{AB_{i}}).$
General poly\-gamy inequalities of multipartite entanglement were also proposed in \cite{K1,K2, JF} for entanglement of assistance.

The generalized monogamy and polygamy relations have an important feature of transitivity. In this regard, the concerned measure also satisfies the monogamy relations for other powers of the measure.  The authors \cite{JFQ} provided a class of monogamy relations of the $\alpha$th $(0\leq\alpha\leq\gamma,\gamma\geq2)$ and polygamy relation for the $\beta$th $(\beta\geq\delta,0\leq\delta\leq1)$ powers for any quantum correlation. Applying the monogamy relations in \cite{JFQ} to quantum correlations like squared convex-roof extended negativity, entanglement of formation, and concurrence one can get tighter monogamy inequalities than those given in \cite{ZF}. In \cite{ZJZ1,ZJZ2}, the authors gave another set of monogamy relations for $(0\leq\alpha\leq\frac{\gamma}{2},\gamma\geq2)$ and polygamy relations for $(\beta\geq\delta,0\leq\delta\leq1)$, the bound is stronger than \cite{JFQ} in the monogamy case but weaker in polygamy case. Recently, the authors \cite{ZLJM} rigorously proved that their monogamy and polygamy are tighter than those given in \cite{JFQ,ZJZ1,ZJZ2}.
 In \cite{LSF, YCFW, TZJF} another type of monogamy relations has also been proposed.
This raises the question of whether there are ways to improve the monogamy and polygamy relations at the same time.

In this study, we will give tighter
monogamy and polygamy relations for the concerned measure by weighted functions. Our finding points out that the weighted monogamy and polygamy relations can achieve significantly better results than previous results using ordinary means of sharpening the inequality. We show that \eqref{e:mono1}
and \eqref{e:poly1} can be simultaneously strengthened by the weighted monogamy and polygamy relations respectively:
\begin{align}\label{e:poly3}
C^{\alpha}(\rho_{A|BC})\geq w_1^{1/\beta}C^{\alpha}(\rho_{AB})+w_2^{1/\beta}C^{\alpha}(\rho_{AC}), \\
C^{\alpha}_{a}(\rho'_{A|BC})\leq w_1^{1/\beta}C^{\alpha}_{a}(\rho'_{AB})+w_2^{1/\beta}C^{\alpha}_{a}(\rho'_{AC}),
\end{align}
where the weights satisfy $\frac{1}{w_1^{1/\alpha}} + \frac{1}{w_2^{1/\alpha}}=1 $ and $\frac{1}{\beta} = \frac{1}{\gamma} - \frac{1}{\alpha} (\gamma \in \mathbb{R}, \gamma \geq 2)$, and $\rho$ and $\rho'$ are two measurements.

Our new method puts both monogamy and polygamy relations in the same setting and we show that the equilibrium occurs exactly at the measurement of the multipartite system.
Our study gives tighter bounds than some of the recent relations
\cite{JFQ,ZJZ1,ZJZ2,ZLJM}. Especially, we show that the weighted monogamy and polygamy are significantly better than the usual monogamy and polygamy (see detailed examples). This reflects the same phenomenon observed by mathematicians working in analysis.

Moreover, we remark that our parameterized bounds (by parameter $s$) in both monogamy and polygamy relations are in fact the {\it strongest} in the sense that whenever another set of bounds is given under the same condition, one can select the parameter $s$ such that the corresponding bounds are tighter than those given.

\section{\textbf{Tighter monogamy/polygamy relations}}

Let $\rho=\rho_{ABC}$ be a state over the Hilbert space $\mathcal{H}_{A}\bigotimes \mathcal{H}_{B}\bigotimes\mathcal{H}_{C}$, and $\mathcal{E}$
a bipartite entanglement measure of quantum correlation. The tripartite state $\rho$ can also be viewed as a bipartite state
if we view it as a state over $\mathcal{H}_{A}\bigotimes (\mathcal{H}_{B}\bigotimes\mathcal{H}_{C})$. The $\gamma$th-monogamy of the measure $\mathcal{E}$ means the following
relation \cite{KPSS}:
\begin{align}\label{e:ABC}
\mathcal{E}^{\gamma}(\rho_{A|BC})\geq \mathcal{E}^{\gamma}(\rho_{AB})+\mathcal{E}^{\gamma}(\rho_{AC}),
\end{align}
where $\rho_{AB}=\operatorname{Tr}_C(\rho)$ (similarly $\rho_{AC}$) is the reduced density matrix.
The usual monogamy refers to the case $\gamma=1$.


It is known that the general $\gamma$th $(\gamma\in \mathbb{R}$, $\gamma\geq 2)$ monogamy of $\mathcal{E}$
holds \cite{OV,BYW,SPSS} for the $N$-partite state $\rho$
\begin{align}\label{e:AB_{i}}
\mathcal{E}^{\gamma}(\rho_{A|B_{1}B_{2}\cdots B_{N-1}})\geq \sum_{i=1}^{N-1} \mathcal{E}^{\gamma}(\rho_{AB_{i}}).
\end{align}

From now on, if the state $\rho$ is clear from the context, we simply denote that $\mathcal{E}(\rho_{AB_{i}})=\mathcal{E}_{AB_{i}}$, $\mathcal{E}(\rho_{A|B_{1}B_{2}\cdots B_{N-1}})=\mathcal{E}_{A|B_{1}B_{2}\cdots B_{N-1}}$.

We would like to address the following three questions:

1) if some $\alpha$th monogamy holds, is it true that all other $\beta$th monogamy holds, and if so, under what conditions?

2) how to quantify the lower bound of $\mathcal E(\rho)$ using the $\alpha$th monogamy?

3) Are there the strongest and tightest bounds of the monogamy and polygamy relations? If yes, in what sense?

To answer these questions, we consider the following function ($t\geq a\geq 1$ are positive parameters)
$$h(x,y)=(1+\frac{a}{y})^{x-1}+(1+\frac{y}{a})^{x-1}t^{x}$$
defined on $(x, y)\in [0, 1]\times (0, \infty)$. To find extreme points, we compute that
\begin{align}\label{e:der}
\frac{\partial h}{\partial y}=(x-1)(y+a)^{x-2}\left(\frac{-a}{y^{x}}+\frac{t^{x}}{a^{x-1}}\right)
\end{align}
then the critical points lie at either
$y=\frac{a}{t}$ or $x=1$. As a function of $y$ (fixing $x$), $h(x,y)$ is increasing on   $(0,\frac{a}{t}]$ and decreasing on $[\frac{a}{t},\infty)$ so the maximum
$h(x, \frac{a}{t})=(1+t)^x$ for fixed $x\in [0, 1]$.
Therefore for $0 \leq x\leq 1$
\begin{align}\label{e:extreme}
(1+t)^{x}\geq \left(1+\frac{a}{s}\right)^{x-1}+\left(1+\frac{s}{a}\right)^{x-1}t^{x}, \quad \forall s>0.
\end{align}

Similarly, consider the function $h(x, y)$ defined on $[1, \infty)\times (0, \infty)$. For fixed $x\geq 1$, the function $h(x, y)$ is decreasing on $(0,\frac{a}{t}]$ and increasing on $[\frac{a}{t},\infty)$, therefore
the minimum is also given by $h(x, \frac{a}{t})=(1+t)^x$
for arbitrary $s>0$.

Therefore we have shown the following result:
\begin{lem}\label{e:1-1}
Let $a\geq 1$ be a real number. Then for $t\geq a\geq 1$ and $0 \leq x\leq 1$
\begin{align}\label{e:mono1a}
(1+t)^{x}\geq\left(1+\frac{a}{s}\right)^{x-1}+\left(1+\frac{s}{a}\right)^{x-1}t^{x}
\end{align}
for any $s>0$. Similarly for $t\geq a\geq 1$ and $x\geq 1$
\begin{align}\label{e:poly1a}
(1+t)^{x}\leq\left(1+\frac{a}{s}\right)^{x-1}+\left(1+\frac{s}{a}\right)^{x-1}t^{x}
\end{align}
for any $s>0$.
\end{lem}

More generally we also have
\begin{lem}\label{l:ineq} Let $p_1\geq \cdots \geq p_N\geq 0$ be $N$ numbers such that $p_{i}\geq a p_{i+1}(i=1,\cdots,N-1)$ for a fixed $a\geq 1$, then for any $s>0$ one has
that
\begin{align}\label{e:1-4}
\left(\sum_{i=1}^{N}p_{i}\right)^{x}\geq \left(1+\frac{a}{s}\right)^{x-1}\sum_{i=1}^{N}\left(\left(1+\frac{s}{a}\right)^{x-1}\right)^{N-i}p_{i}^{x}
\end{align}
for $0\leq x \leq 1$.
\end{lem}
\begin{proof} We use induction on $N$. The case of $N = 1$ holds by $0<\left(1+\frac{a}{s}\right)^{x-1}\leq 1$.
Consider $N$ numbers $p_i\geq 0$ such that $p_i\geq a p_{i+1}$ for a fixed $a\geq 1$.
We can assume $p_{N}\neq0$ otherwise the inequality holds by the inductive hypothesis. Note that $p_{1}+p_{2}+\cdots+p_{N-1}>a{p_{N}}$, so
Lemma \ref{e:1-1} implies that
\begin{align*}
\left(\sum_{i=1}^{N}p_{i}\right)^{x} &=p_{N}^{x}\left(1+\frac{p_{1}+p_{2}+\cdots+p_{N-1}}{p_{N}}\right)^{x}\\
&\geq p_{N}^{x}\left(1+\frac{a}{s}\right)^{x-1}+\left(1+\frac{s}{a}\right)^{x-1}\left(\frac{p_{1}+p_{2}+\cdots+p_{N-1}}{p_{N}}\right)^{x}.
\end{align*}
Using the inductive hypothesis, the above is no less than the following
\begin{align*}
&\geq \left(1+\frac{a}{s}\right)^{x-1}\sum_{i=2}^{N}\left(\left(1+\frac{s}{a}\right)^{x-1}\right)^{N-i}p_{i}^{x}+\left(\left(1+\frac{s}{a}\right)^{x-1}\right)^{N-1}p_{1}^{x}\\
&\geq\left(1+\frac{a}{s}\right)^{x-1}\sum_{i=1}^{N}\left(\left(1+\frac{s}{a}\right)^{x-1}\right)^{N-i}p_{i}^{x}
\end{align*}
where the last inequality has used $0<\left(1+\frac{a}{s}\right)^{x-1}\leq 1$.
\end{proof}

Lemma \ref{e:1-1} and its proof (in particular for $\frac{a}{t}\leq s \leq 1$) implies that
\begin{equation}\label{e:1-3}
\begin{aligned}
(1+t)^{x}&\geq\left(1+\frac{a}{s}\right)^{x-1}+\left(1+\frac{s}{a}\right)^{x-1}t^{x}\\
&\geq(1+a)^{x-1}+\left(1+\frac{1}{a}\right)^{x-1}t^{x}.
\end{aligned}
\end{equation}

As for $N$-partite case,
to use Lemma
\ref{l:ineq}, we let
$$p=max\left\{\frac{ap_{2}}{p_{1}},\frac{ap_{3}}{p_{1}+p_{2}},\cdots,\frac{ap_{N}}{p_{1}+\cdots+p_{N-1}}\right\}.$$

If $p\leq s \leq 1$, then we have
\begin{equation}\label{e:tighter relation}
\begin{aligned}
\left(\sum_{i=1}^{N}p_{i}\right)^{x}&\geq \left(1+\frac{a}{s}\right)^{x-1}\sum_{i=1}^{N}\left(\left(1+\frac{s}{a}\right)^{x-1}\right)^{N-i}p_{i}^{x}\\
&\geq\left(1+a\right)^{x-1}\sum_{i=1}^{N}\left(\left(1+\frac{1}{a}\right)^{x-1}\right)^{N-i}p_{i}^{x}
\end{aligned}
\end{equation}
Naturally, our relations based on Lemma \ref{e:1-1} will outperform those given in \cite[Lem. 2]{ZLJM}, and subsequently also
\cite{JFQ,ZJZ1,ZJZ2}.

\bigskip
Now let's first utilize the lemma for the monogamy relation.
\begin{thm}\label{t:1-1}
Let $\mathcal{E}$ be a bipartite quantum measure satisfying the generalized
monogamy relation \eqref{e:ABC} for the tripartite state $\rho_{ABC}$ with $\gamma\geq 2$.
Suppose $\mathcal{E}_{AB}^{\gamma}\geq a\mathcal{E}_{AC}^{\gamma}$ for some $a\geq 1$, then for any $0\leq \alpha\leq \gamma$ and $s>0$, one has that
\begin{align}\label{e:mono1b}
\mathcal{E}_{A|BC}^{\alpha}\geq\left(1+\frac{a}{s}\right)^{\frac{\alpha}{\gamma}-1}\mathcal{E}_{AC}^{\alpha}+\left(1+\frac{s}{a}\right)^{\frac{\alpha}{\gamma}-1}\mathcal{E}_{AB}^{\alpha}.
\end{align}
Moreover, we claim that the lower bound on the right is the strongest and tightest among all peers.
\end{thm}
\begin{proof} It follows from \eqref{e:ABC} that
\begin{align*}
\mathcal{E}_{A|BC}^{\alpha}=(\mathcal{E}_{A|BC}^{\gamma})^{\frac{\alpha}{\gamma}}&\geq (\mathcal{E}^{\gamma}_{AB}+\mathcal{E}^{\gamma}_{AC})^{\frac{\alpha}{\gamma}}\\
&=\mathcal{E}_{AC}^{\alpha}\left(1+\frac{\mathcal{E}^{\gamma}_{AB}}{\mathcal{E}^{\gamma}_{AC}}\right)^{\frac{\alpha}{\gamma}}\\
&\geq\left(1+\frac{a}{s}\right)^{\frac{\alpha}{\gamma}-1}\mathcal{E}_{AC}^{\alpha}+\left(1+\frac{s}{a}\right)^{\frac{\alpha}{\gamma}-1}\mathcal{E}_{AB}^{\alpha}
\end{align*}

To see why our bound is the tightest.
Observe that the right-hand side of \eqref{e:mono1b} produces a family of lower bounds indexed by $s$. Observe that
 \begin{equation}\label{e:reason}
 \lim_{s\to \frac{a}{t}}(1+\frac{a}{s})^{x-1}+(1+\frac{s}{a})^{x-1}t^x=(1+t)^x
 \end{equation}
 and $(1+t)^x$ is the maximum value of $h(x, s)$
of $\mathcal E_{A|BC}$ is given, we can always find a parameter $s\in [\frac{a}{t}, 1]$ such that
 our bound in \eqref{e:mono1b} is no less than the given one. In other words, the bounds given above will be the strongest and tightest bounds among all lower bounds.
 \end{proof}

The general monogamy inequality of the $\alpha$th power of the quantum measure for $N$-qubit quantum states is given by the following theorem.
\begin{thm} \label{t:1-2}
Let $\mathcal{E}$ be a bipartite quantum measure satisfying the general $\gamma$th-monogamy relation \eqref{e:AB_{i}} for $\gamma\geq 2$ and $\rho_{AB_{1}\cdots B_{N-1}}$ be any $N$-qubit quantum states. Arrange \{$\mathcal{E}_{i}=\mathcal{E}_{AB_{j}}|i,j=1,\cdots,N-1\}$ in descending order.  If $a\geq1$ and  $\mathcal{E}_{i}^{\gamma}\geq a\mathcal{E}_{i+1}^{\gamma}$ for $i=1,\cdots,N-2$, then
\begin{align}
\mathcal{E}_{A|B_{1}\cdots B_{N-1}}^{\alpha}\geq \left(1+\frac{a}{s}\right)^{\frac{\alpha}{\gamma}-1}\sum_{i=1}^{N-1}\left(\left(1+\frac{s}{a}\right)^{\frac{\alpha}{\gamma}-1}\right)^{N-1-i}\mathcal{E}_{i}^{\alpha}
\end{align}
for $0\leq \alpha\leq \gamma$ and arbitrary $s>0$.
\end{thm}
\begin{proof} Similar to Theorem \ref{t:1-1}. It follows from the generalized monogamy relation \eqref{e:AB_{i}} and the inequality \eqref{e:1-4}.
\end{proof}

{\bf Comparison of new monogamy relations with previous works}.
Selecting an appropriate parameter $s$, it follows from \eqref{e:tighter relation} and Theorem \ref{t:1-2} that
\begin{align*}
\mathcal{E}_{A|B_{1}\cdots B_{N-1}}^{\alpha}&\geq \left(1+\frac{a}{s}\right)^{\frac{\alpha}{\gamma}-1}\sum_{i=1}^{N-1}\left(\left(1+\frac{s}{a}\right)^{\frac{\alpha}{\gamma}-1}\right)^{N-1-i}\mathcal{E}_{i}^{\alpha}\\
&\geq \left(1+a\right)^{\frac{\alpha}{\gamma}-1}\sum_{i=1}^{N-1}\left(\left(1+\frac{1}{a}\right)^{\frac{\alpha}{\gamma}-1}\right)^{N-1-i}\mathcal{E}_{i}^{\alpha},
\end{align*}
where the middle term corresponds to our {\it lower bound} and the last term was that of 
\cite[Thm. 2]{ZLJM}.
 Since \cite{ZLJM} provided stronger bounds than \cite{JFQ,ZJZ1,ZJZ2}, any monogamy relations based on Lemma \ref{e:1-1} will produce tighter bounds than those given in \cite{JFQ,ZJZ1,ZJZ2}
 as well.

\bigskip
We have remarked that polygamy relations are really the dual forms of monogamy relations.
It is known that
for arbitrary dimensional tripartite state there exists $0\leq \delta\leq 1$ such that any quantum correlation measure $\mathcal{E}$ satisfies the following polygamy relation \cite{JF}:
\begin{align}\label{e:polygamy-1}
\mathcal{E}^{\delta}(\rho_{A|BC})\leq \mathcal{E}^{\delta}(\rho_{AB})+\mathcal{E}^{\delta}(\rho_{AC})
\end{align}
One can easily derive the following generalized polygamy relation for arbitrary dimensional $N$-qubit states
\begin{align}\label{e:polygamy-2}
\mathcal{E}^{\delta}(\rho_{A|B_{1}\cdots B_{N-1}})\leq \sum_{i=1}^{N-1}\mathcal{E}^{\delta}(\rho_{AB_{i}}).
\end{align}

If we use Lemma \ref{e:1-1} for the polygamy relation, a similar argument as above will immediately give the following result.

\begin{thm}\label{t:2-1}
Let $\mathcal{E}$ be a bipartite quantum measure satisfying the $\delta$th
polygamy relation \eqref{e:polygamy-1} for a tripartite state $\rho_{ABC}$ and $0\leq\delta\leq 1$.
Suppose $\mathcal{E}_{AB}^{\delta}\geq a\mathcal{E}_{AC}^{\delta}$ for some $a\geq1$, then for any $\beta\geq \delta$ and $s>0$, we have that
\begin{align}\label{e:polyb}
\mathcal{E}_{A|BC}^{\beta}\leq\left(1+\frac{a}{s}\right)^{\frac{\beta}{\delta}-1}\mathcal{E}_{AC}^{\beta}+\left(1+\frac{s}{a}\right)^{\frac{\beta}{\delta}-1}\mathcal{E}_{AB}^{\beta}.
\end{align}
\end{thm}

\begin{thm} \label{t:2-2}
Let $\mathcal{E}$ be a bipartite quantum measure satisfying the generalized polygamy relation \eqref{e:polygamy-2} for $0\leq\delta\leq 1$ and $\rho_{AB_{1}\cdots B_{N-1}}$ be any $N$-qubit quantum states. Arrange \{$\mathcal{E}_{i}=\mathcal{E}_{AB_{j}}|i,j=1,\cdots,N-1\}$ in descending order.  If $a\geq1$ and $\mathcal{E}_{i}^{\delta}\geq a\mathcal{E}_{i+1}^{\delta}$ for $i=1,\cdots,N-2$, then
\begin{align}\label{e:2-4}
\mathcal{E}_{A|B_{1}\cdots B_{N-1}}^{\beta}\leq \left(1+\frac{a}{s}\right)^{\frac{\beta}{\delta}-1}\sum_{i=1}^{N-1}\left(\left(1+\frac{s}{a}\right)^{\frac{\beta}{\delta}-1}\right)^{N-1-i}\mathcal{E}_{i}^{\beta}
\end{align}
for $\beta\geq \delta$ and arbitrary $s>0$.
\end{thm}

{\bf Comparison of our polygamy relations with previous works}.
Consider $\frac{a}{t}\leq s \leq 1$ and $x \geq 1$, by Lemma \ref{e:1-1} and its proof, we have $$(1+t)^{x}\leq\left(1+\frac{a}{s}\right)^{x-1}+\left(1+\frac{s}{a}\right)^{x-1}t^{x}\leq(1+a)^{x-1}+\left(1+\frac{1}{a}\right)^{x-1}t^{x}.$$
Therefore our {\it upper} bound of $(1+t)^{x}$ is better than that of \cite[Lem. 3]{ZLJM}. The authors \cite{ZLJM} have proved in detail that their results are stronger than \cite{JFQ,ZJZ1,ZJZ2}, thus any polygamy relations derived from our Lemma \ref{e:1-1} are tighter than those given in \cite{JFQ,ZJZ1,ZJZ2}.

We remark that due to \eqref{e:reason}
 our bound produced by \eqref{e:polyb} will be the tightest. This means that whenever an upper bound
 of $\mathcal E_{A|BC}$ is given, we can always find a parameter $s\in [\frac{a}{t}, 1]$ such that
 our bound in \eqref{e:polyb} is no bigger than the given one.

When using Theorem \ref{t:2-2}, let
$$r=max\{\frac{a\mathcal{E}_{2}^{\delta}}{\mathcal{E}_{1}^{\delta}},\frac{a\mathcal{E}_{3}^{\delta}}{\mathcal{E}_{1}^{\delta}+\mathcal{E}_{2}^{\delta}},\cdots,\frac{a\mathcal{E}_{N-1}^{\delta}}{\mathcal{E}_{1}^{\delta}+\cdots+\mathcal{E}_{N-2}^{\delta}}\}$$
Let $r\leq s \leq 1$, then it follows from
Theorem \ref{t:2-2} that
\begin{align*}
\mathcal{E}_{A|B_{1}\cdots B_{N-1}}^{\beta}&\leq \left(1+\frac{a}{s}\right)^{\frac{\beta}{\delta}-1}\sum_{i=1}^{N-1}\left(\left(1+\frac{s}{a}\right)^{\frac{\beta}{\delta}-1}\right)^{N-1-i}\mathcal{E}_{i}^{\beta}\\
&\leq \left(1+a\right)^{\frac{\beta}{\delta}-1}\sum_{i=1}^{N-1}\left(\left(1+\frac{1}{a}\right)^{\frac{\beta}{\delta}-1}\right)^{N-1-i}\mathcal{E}_{i}^{\beta}
\end{align*}
This means the conclusion \eqref{e:2-4} of our result is better than \cite[Thm. 4]{ZLJM} for $r\leq s \leq 1$. Therefore our bounds produced are tighter polygamy relations than those in \cite{JFQ,ZJZ1,ZJZ2}.

\bigskip
Now let's use examples to show why our bounds in both monogamy and polygamy relations are the strongest among all similar bounds.

The generalized monogamy relations are applicable to any quantum correlation measure, such as negativity, entanglement of formation, concurrence, etc. We take the concurrence to demonstrate
the advantages of our monogamy relations.

Recall that the concurrence of a pure
state $\rho$ on $\mathcal{H}_A \otimes \mathcal{H}_B$ is defined by \cite{AF,RBCGM,U} $C\left(|\psi\rangle_{A B}\right)=\sqrt{2\left[1-\operatorname{Tr}\left(\rho_A^2\right)\right]}$, where $\rho_A$ is the reduced density matrix.
For a mixed state $\rho_{A B}$ the concurrence is given by
$$C\left(\rho_{A B}\right)=\min _{\left\{p_i,\left|\psi_i\right\rangle\right\}} \sum_i p_i C\left(\left|\psi_i\right\rangle\right),
$$
where the minimum is taken over all possible decompositions of $\rho_{A B}=\sum_i p_i\left|\psi_i\right\rangle\left\langle\psi_i\right|$ with $p_i \geqslant 0, \sum_i p_i=1$ and $\left|\psi_i\right\rangle \in \mathcal{H}_A \otimes \mathcal{H}_B$.

\begin{exmp} Let $\rho=|\psi\rangle\langle\psi|$ be the three-qubit state \cite{AACJLT}:
$$
|\psi\rangle=\frac1{2}(|000\rangle+|110\rangle)+\frac{\sqrt{6}}6(e^{i \varphi}|100\rangle+|101\rangle+|111\rangle).
$$
Then the concurrence $C_{A\mid BC}=\frac{\sqrt{21}}{6}, C_{AB }=\frac{\sqrt{6}}{6}, C_{AC}=\frac{1}{2}$. Thus $t=\frac{C_{AC}^{\gamma}}{C_{AB}^{\gamma}}=(\frac{\sqrt{6} }{2})^{\gamma}$. Set $a = 1.05^{\frac{\gamma}{2}}$ and $s = 0.72^{\frac{\gamma}{2}}$ (satisfying $t \geq a \geq 1$ and $s \in [\frac{a}{t},1])$.
By Theorem \ref{t:1-1} our lower bound is
\begin{align*}
Z_1&=\left(1+\frac{a}{s}\right)^{\frac{\alpha}{\gamma}-1} C_{AB}^\alpha+\left(1+\frac{s}{a}\right)^{\frac{\alpha}{\gamma}-1} C_{A C}^\alpha\\
&=\left(1+\frac{ 1.05^{\frac{\gamma}{2}}}{ 0.72^{\frac{\gamma}{2}}}\right)^{\frac{\alpha}{\gamma}-1}\left(\frac{\sqrt{6}}{6}\right)^\alpha+\left(1+\frac{ 0.72^{\frac{\gamma}{2}}}{ 1.05^{\frac{\gamma}{2}}}\right)^{\frac{\alpha}{\gamma}-1}\left(\frac{1}{2}\right)^\alpha
\end{align*}

The following lower bound $Z_2$ comes from \cite{ZLJM}, which is a special case of our bound at $s=1$.
\begin{align*}
    Z_2&=\left(1+a\right)^{\frac{\alpha}{\gamma}-1} C_{AB}^\alpha+\left(1+\frac{1}{a}\right)^{\frac{\alpha}{\gamma}-1} C_{A C}^\alpha\\
&=\left(1+1.05^{\frac{\gamma}{2}}\right)^{\frac{\alpha}{\gamma}-1}\left(\frac{\sqrt{6}}{6}\right)^\alpha+\left(1+\frac{1}{1.05^{\frac{\gamma}{2}}}\right)^{\frac{\alpha}{\gamma}-1}\left(\frac{1}{2}\right)^\alpha
\end{align*}

The lower bound given in \cite{JFQ} is
\begin{equation*}
    Z_{_3}= C_{A B}^{\alpha} + \frac{(1+a)^{\frac{\alpha}{\gamma}} - 1}{a^{\frac{\alpha}{\gamma}}} C_{AC}^{\alpha}
    = \left(\frac{\sqrt{6}}{6}\right)^{\alpha}+ \frac{(1+1.05^{\frac{\gamma}{2}})^{\frac{\alpha}{\gamma}} - 1}{1.05^{\frac{\alpha}{2}}}\left(\frac{1}{2}\right)^{\alpha}
\end{equation*}

Let $p=\frac{1}{2}$, the lower bound given in \cite{ZJZ1,ZJZ2} is
\begin{align*}
    Z_{_4}&=p^{\frac{\alpha}{\gamma}}C_{AB}^{\alpha} + \frac{(1+a)^{\frac{\alpha}{\gamma}} - p^{\frac{\alpha}{\gamma}}}{a^{\frac{\alpha}{\gamma}}}C_{AC}^{\alpha} \\ &= \left(\frac{1}{2}\right)^{\frac{\alpha}{\gamma}}\left(\frac{\sqrt{6}}{6}\right)^{\alpha}+ \frac{(1+1.05^{\frac{\gamma}{2}})^{\frac{\alpha}{\gamma}} - \left(\frac{1}{2}\right)^{\frac{\alpha}{\gamma}}}{1.05^{\frac{\alpha}{2}}}\left(\frac{1}{2}\right)^{\alpha}
\end{align*}

\begin{figure}[H]
\centering
\includegraphics[scale=0.55]{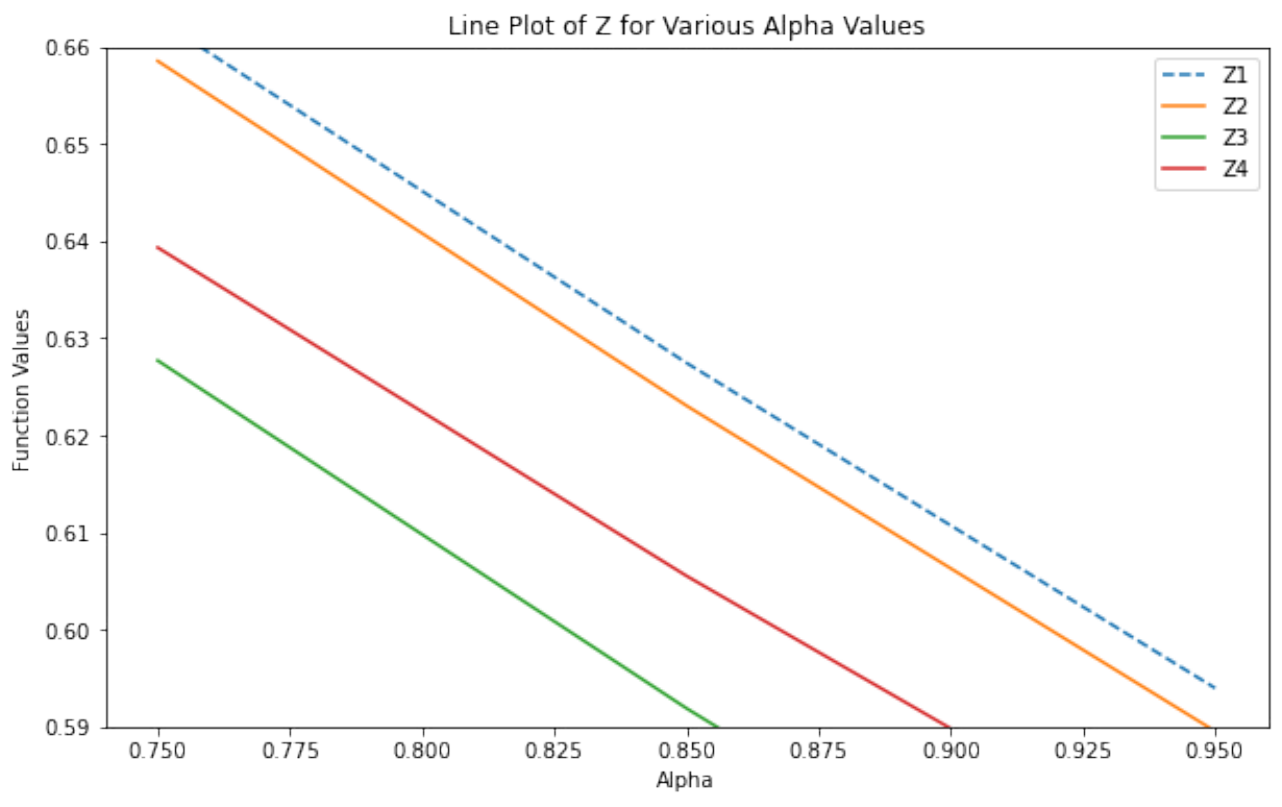}
\end{figure}

Figure 1: The $x$-axis is the exponent $\alpha$, and the $y$-axis shows the lower bounds of the concurrence $C_{A \mid BC}^\alpha$ (for $\gamma=3$). The dashed blue line $Z_1$ represents our lower bound at $s=0.72^{1.5}$ and the orange line $Z_2$ depicts the lower bound from \cite{ZLJM}. The green line $Z_3$ plots the lower bound from \cite{JFQ}. The red line $Z_4$ is the lower bound from \cite{ZJZ1,ZJZ2}. The graph shows our lower bound $Z_1$ is the strongest.
\end{exmp}

\bigskip

Now let's consider the polygamy relation.
The generalized polygamy relations work for many quantum correlation measures such as the concurrence of assistance, the entanglement of assistance, and square of convex-roof extended negativity of assistance (SCRENoA), etc. Correspondingly, a new class of polygamy relations are obtained. We will take the SCRENoA as an example.

The negativity of bipartite state $\rho_{A B}\in \mathcal{H}_{A}\bigotimes \mathcal{H}_{B}$ is defined by \cite{VW}: $$N\left(\rho_{AB}\right)=\left(\left\|\rho_{A B}^{T_{A}}\right\|-1\right) / 2,$$ where $\rho_{AB}^{T_{A}}$ is the partial transposition with respect to the subsystem $A$ and $\|X\|=\operatorname{Tr} \sqrt{X X^{\dagger}}$ is the trace norm of $X$. For simplicity, we define the negativity: $N\left(\rho_{AB}\right)=\left\|\rho_{AB}^{T_{A}}\right\|-1$. The negativity $N\left(\rho_{AB}\right)$ for any bipartite pure state $|\psi\rangle_{AB}$ is given by
$$
N\left(|\psi\rangle_{AB}\right)=2 \sum_{i<j} \sqrt{\lambda_i \lambda_j}=\left(\operatorname{Tr} \sqrt{\rho_{A}}\right)^2-1,
$$
where $\lambda_i$ are the eigenvalues for the reduced density matrix $\rho_A$ of $|\psi\rangle_{AB}$.

For a mixed state $\rho_{AB}$, the square of convex-roof extended negativity (SCREN) is defined by $N_{s c}\left(\rho_{A B}\right)=\left[\min \sum_i p_i N\left(\left|\psi_i\right\rangle_{AB}\right)\right]^2$, where the minimum is taken over all possible pure state decompositions $\left\{p_i,\left|\psi_i\right\rangle_{AB}\right\}$ of $\rho_{AB}$. The SCRENoA is then defined by $N_{s c}^a\left(\rho_{AB}\right)=\left[\max \sum_i p_i N\left(\left|\psi_i\right\rangle_{AB}\right)\right]^2$, where the maximum is taken over all possible pure state decompositions $\left\{p_i,\left|\psi_i\right\rangle_{AB}\right\}$ of $\rho_{AB}$. For convenience, we abbreviate: $N_{a _{AB_i}}=N_{s c}^a\left(\rho_{AB_i}\right)$ and $N_{a _{A \mid B_1 \cdots B_{N-1}}}=N_{s c}^a\left(|\psi\rangle_{A \mid B_1 \cdots B_{N-1}}\right)$.
\begin{cor}\label{c:2-1}
 Let $\delta \in(0,1)$ be the fixed number so that the SCRENoA satisfies the generalized polygamy relation \eqref{e:polygamy-1}. Suppose $N_{a_{ AC}}^{\delta} \geqslant a N_{a _{AB}}^{\delta}$ for $a \geq 1$, then the SCRENoA satisfies
\begin{equation}\label{e:2-2}
N_{a_{A \mid BC}}^{\beta} \leq\left(1+\frac{a}{s}\right)^{\frac{\beta}{\delta}-1} N_{a _{AB}}^{\beta}+\left(1+\frac{s}{a}\right)^{\frac{\beta}{\delta}-1} N_{a _{AC}}^\beta
\end{equation}
for arbitrary $\beta \geq \delta$ and $s>0$.

Moreover, the upper bound on the right is the strongest and tightest among all bounds.
\end{cor}
By inequality \eqref{e:2-2} and induction, the following result is immediate for an $N$-qubit quantum state $\rho_{AB_{1} \cdots B_{N-1}}$.

\begin{cor} \label{c:2-2}
Let $\delta \in(0,1)$ be the fixed number so that the SCRENoA satisfies the generalized polygamy relation \eqref{e:polygamy-2}. If $\{N_{a_{i}}=N_{a_{AB_{j}}}|i,j=1,\cdots,N-1\}$ is arranged in descending order.  If $N_{a_{i}}^{\delta}\geq aN_{a_{i+1}}^{\delta}$ for all $i=1,\cdots,N-2$ and a fixed $a\geq 1$, then
\begin{align}
N_{a _{A \mid B_1 \ldots B_{N-1}}}^\beta \leq\left(1+\frac{a}{s}\right)^{\frac{\beta}{\delta}-1} \sum_{i=1}^{N-1}\left(\left(1+\frac{s}{a}\right)^{\frac{\beta}{\delta}-1}\right)^{N-1-i} N_{a_{i}}^\beta
\end{align}
for any $\beta \geq \delta$ and $s>0$.
\end{cor}

We remark that our result
is better than \cite[Cor. 4]{ZLJM} for $r\leq s \leq 1$. Subsequently, our bounds give tighter polygamy relations than those in \cite{JFQ,ZJZ1,ZJZ2}.

\begin{exmp} Let us consider the three-qubit generlized $W$-class state,
$$
|W\rangle_{ABC}=\frac{1}{2}(|100\rangle+|010\rangle)+\frac{\sqrt{2}}{2}|001\rangle.
$$
Then $N_{a_{ A\mid BC}}=\frac{3}{4}, N_{a_{ AB}}=\frac{1}{4}, N_{a_{ AC}}=\frac{1}{2}$. Note that $\beta\geq\delta,0\leq\delta\leq1$, we take $\delta = 0.6$, $a=1.2$. Then $t=\frac{N_{a_{ AC}}^{\delta}}{N_{a_{ AB}}^{\delta}}=2^{0.6}$.

Taking
$s=\frac{a+t}{2t} \in[\frac{a}{t}, 1]$,
Corollary \ref{c:2-1} says that our upper bound is
\begin{align*}
W_{_1}&=\left(1+\frac{a}{s}\right)^{\frac{\beta}{\delta}-1} N_{a _{AB}}^{\beta}+\left(1+\frac{s}{a}\right)^{\frac{\beta}{\delta}-1} N_{a _{AC}}^{\beta}\\
&=\left(1+\frac{1.2}{s}\right)^{\frac{\beta}{0.6}-1} \left(\frac{1}{4}\right)^{\beta}+\left(1+\frac{s}{1.2}\right)^{\frac{\beta}{0.6}-1} \left(\frac{1}{2}\right)^{\beta}
\end{align*}

The upper bound given in \cite{ZLJM} is
\begin{align*}
W_{_2}&=\left(1+a\right)^{\frac{\beta}{\delta}-1} N_{a _{AB}}^{\beta}+\left(1+\frac{1}{a}\right)^{\frac{\beta}{\delta}-1} N_{a _{AC}}^{\beta}\\
&=\left(1+2^{0.6}\right)^{\frac{\beta}{0.6}-1} \left(\frac{1}{4}\right)^{\beta}+\left(1+(2)^{-0.6}\right)^{\frac{\beta}{0.6}-1} \left(\frac{1}{2}\right)^{\beta}
\end{align*}

The upper bound given in \cite{JFQ} is
\begin{equation*}
    W_{_3}= N_{a _{AB}}^{\beta} + \frac{(1+a)^{\frac{\beta}{\delta}} - 1}{a^{\frac{\beta}{\delta}}} N_{a _{AC}}^{\beta}
    = \left(\frac{1}{4}\right)^{\beta}+ \frac{(1+1.2)^{\frac{\beta}{0.6}} - 1}{1.2^{\frac{\beta}{0.6}}}\left(\frac{1}{2}\right)^{\beta}
\end{equation*}

The upper bound given in \cite{ZJZ1,ZJZ2} is
\begin{align*}
    W_{_4}&=p^{\frac{\beta}{\delta}}N_{a _{AB}}^{\beta} + \frac{(1+a)^{\frac{\beta}{\delta}} - p^{\frac{\beta}{\delta}}}{a^{\frac{\beta}{\delta}}}N_{a _{AC}}^{\beta} \\ &= \left(\frac{1}{2}\right)^{\frac{\beta}{0.6}}\left(\frac{1}{4}\right)^{\beta}+ \frac{(1+1.2)^{\frac{\beta}{0.6}} - \left(\frac{1}{2}\right)^{\frac{\beta}{0.6}}}{1.2^{\frac{\beta}{0.6}}}\left(\frac{1}{2}\right)^{\beta}
\end{align*}

\begin{figure}[H]
\centering
\includegraphics[scale=0.55]{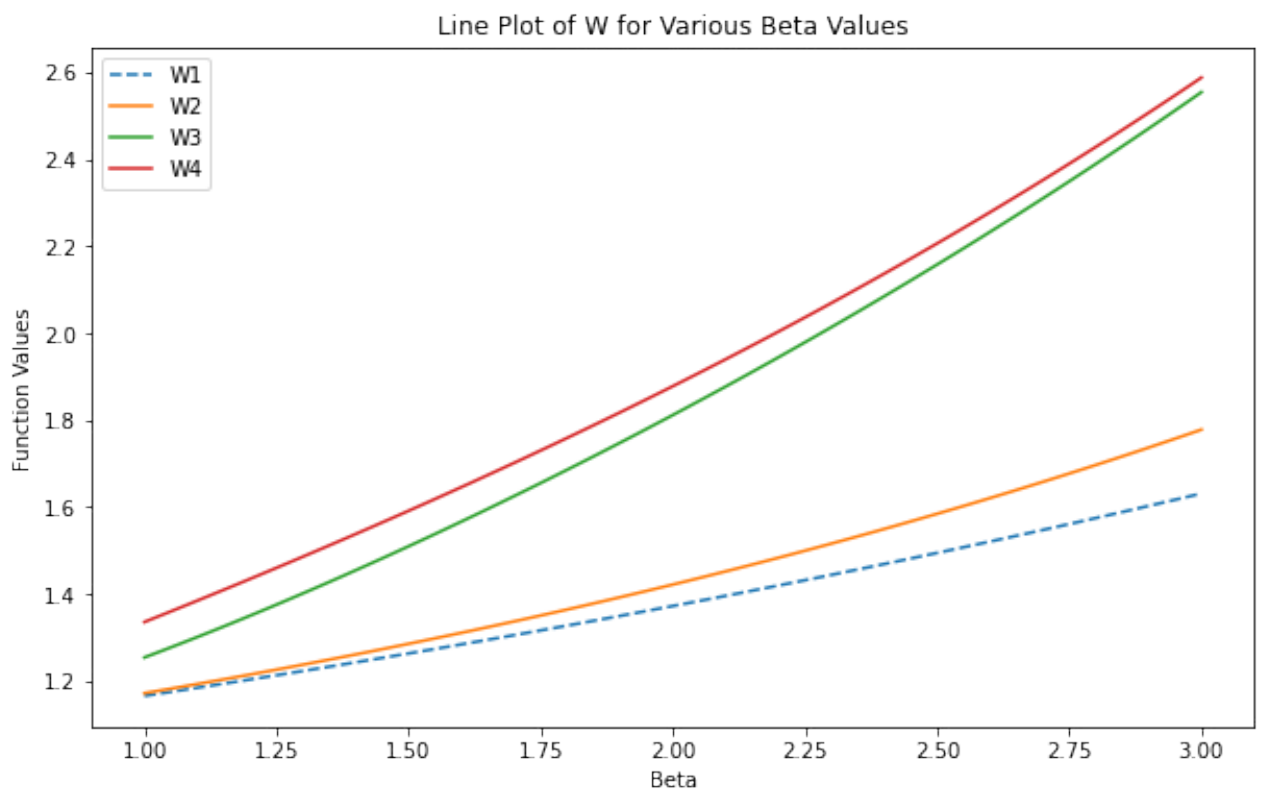}
\end{figure}
Figure 2: The $y$-axis represents SCRENoA for the state $|W\rangle_{ABC}$ as a function of $\beta$ when $\delta = 0.6$. The dashed blue line $W_1$ represents our upper bound ($a=1.2, s=1.2\cdot 2^{-1.6}+0.5$). The orange line $W_2$ plots the upper bound from \cite{ZLJM}. The green line $W_3$ indicates the upper bound from \cite{JFQ}. The red line $W_4$ depicts the upper bound from \cite{ZJZ1,ZJZ2}. It's clear that our bound $W_1$ is the tightest.
\end{exmp}

\section{\textbf{Conclusion}}
Our study on the monogamy and polygamy relations related to quantum correlations in multipartite quantum systems is the first among its peers. By introducing a family of weighted monogamy and polygamy relations indexed by $s$, we have obtained the {\it strongest and tightest bounds} for general
monogamy and polygamy relations for any quantum measurement. We illustrate in detail for parameter $s\in [\frac{a}{t},  1]$, the monogamy and polygamy relations of quantum correlations for the cases $(0\leq\alpha\leq\gamma,\gamma\geq2)$ and $(\beta\geq\delta,0\leq\delta\leq1)$ are tighter than some of the recently available bounds. Taking the concurrence and the SCRENoA (square of convex-roof extended negativity of assistance) as detailed examples, we also verify that our bounds are indeed the {\it tightest} on graphs in both cases.

In summary, we have shown that our (weighted) bounds provide the {\it best and tightest} bounds for the monogamy ($\alpha\leq\gamma$) and polygamy relations ($\beta\geq\gamma$) in their kinds.

\bibliographystyle{plain}

\end{document}